\documentclass[letterpaper, 10 pt, conference]{ieeeconf}  

\IEEEoverridecommandlockouts                              

\usepackage{comment}
\usepackage{cite} 
\let\labelindent\relax
\usepackage{enumitem}

\usepackage{amssymb, mathbbol, mathtools, amsthm} 
\usepackage{csquotes} 
\usepackage{hyperref} 
\usepackage{xcolor} 
\usepackage{textcomp}
\usepackage[export]{adjustbox}
\usepackage{subcaption}
\usepackage{epstopdf}
\usepackage{setspace,tikz,booktabs,multirow}
\usepackage{float}
\usepackage{array} 
\usepackage{graphics,ragged2e} 
\usepackage{epsfig} 
\usepackage{caption}
\usepackage{mathptmx} 
\usepackage{times} 
\IEEEoverridecommandlockouts                              

\newcounter{remark}
\renewcommand{\theremark}{\arabic{remark}}
\newenvironment{remark}[1][]{%
  \refstepcounter{remark}%
  \par\noindent\textbf{Remark \theremark. #1} \itshape
}{\par}
\newtheorem{theorem}{Theorem}
\newtheorem{lemma}{Lemma}
\newtheorem{definition}{Definition}
\newtheorem{theory}{Theory}
\newtheorem{example}{Example}
\theoremstyle{assumption}
\newtheorem{assumption}{Assumption}
\theoremstyle{proposition}

\theoremstyle{corollary}

\theoremstyle{definition}
\title{\LARGE \bf
How Do Emotions Spread And Affect Decision-Making In Social Networks? An Opinion Dynamics Perspective
}
\author{Aditi Saxena$^{1}$ and Twinkle Tripathy$^{2}$
\thanks{$^{1}$Aditi Saxena is a doctoral student and $^{2}$Twinkle Tripathy is an Assistant Professor in the Department of Electrical Engineering, Indian Institute of Technology Kanpur, India. (email: {\tt\small saditi23@iitk.ac.in}, {\tt\small ttripathy@iitk.ac.in})}}
\begin{document}
\maketitle
\thispagestyle{empty}
\pagestyle{empty}
\begin{abstract}
\emph{Emotions} affect how we think and perceive information. 
Motivated by this, we present a novel framework for modeling \emph{emotions}, and study its impact in social networks. We represent \emph{emotions} on a circle by utilizing the \emph{circumplex model}. These emotions spread in a group through the process of  \emph{emotional contagion}, and influence the interpersonal relationships therein. We propose a model to explain the process of \emph{emotional contagion}, and its effect on the decision-making process, using a Laplacian-based flow. 
Our key finding in this direction shows that agents with similar \emph{emotions} reach agreement eventually, exhibiting \emph{emotional contagion}. Through this, the widely observed behaviors of consensus, polarization, and clustering of opinions follow in complex-weighted networks. Our approach relies on the equitable partitions and spectral properties of Laplacians in complex-weighted networks. Lastly, we validate our findings using synthetic and real-world datasets.
\end{abstract}
\vspace{-3mm}
\section{Introduction}
\label{sec:introduction}
Emotion is an individual’s first response to an event, which plays a vital role in the process of decision-making. When agents in a social network begin to interact, the process of collective decision-making is shaped by two fundamental factors. First, the nature of interactions among agents which can be either cooperative or antagonistic. Second, and equally important, is the emotional state of the agents themselves. Experimental research by psychologists demonstrates that emotions affect information dissemination, and shape the way of processing information \cite{2015emotionanddecisionmaking,2020agentbasedmodel}. Moreover, agents tend to align their emotional states with those of like-minded agents during social interactions in a social network \cite{herrando2021emotional}. 

The rapidly evolving field of opinion dynamics has a substantial amount of work on the decision-making processes. While the existing models, such as DeGroot's, Friedkin-Johnsen, and the bounded confidence model capture behaviors such as consensus, stubbornness, and polarization \cite{PROSKURNIKOV2018166}; but the effect of emotions is not considered in the opinion evolution process. Hence, the modeling of emotion-affected decision-making is largely unexplored. Although a recent work on the evolution of sentiments investigates how sentiments evolve on \emph{Weibo} (China’s microblogging platform) by using the existing opinion-dynamics models\cite{anton2025opiniondynamicsmodelssentiment}.

Motivated by the pivotal role of \emph{emotions} in shaping decisions, we develop a mathematical framework to model emotions and investigate their affect in signed networks. We encapsulate the renowned \textit{circumplex model} 
to represent emotions on a circle in the complex plane\cite{russell1980circumplex}. We propose a model to show the diffusion of {emotions} by representing them using complex-weighted networks. Complex-weighted networks provide a versatile framework for analyzing various physical and computational systems in the areas of quantum communication, machine learning, electrodynamics and social sciences. These complex-weighted networks may correspond to Hermitian or non-Hermitian matrices, each offering distinct insights into stability, and spectral behavior \cite{bottcher2024complex}. Most of the works on complex-weighted networks like \cite{estrada2025geometries,tian2024structuralbalance} consider Hermitian Laplacian matrices for their analyses. 
Some recent works like \cite{DONG20161,2023structuralbalance} introduce balance conditions and study consensus protocols in non-Hermitian settings\cite{saxena2024reep,saxena2024flows,saxena2025laplacian}. The closest work to our present work on complex-weighted networks is \cite{2023structuralbalance},
where structural balance 
 is analyzed along with \emph{multi-partite consensus} in complex-weighted networks. In our work, we use a different definition of Laplacian matrix (mentioned in Sec.\ref{sec:prelims}) and allow the weights of the network to be complex-weighted within and among the clusters rather than having only co-operative relations within the cluster as in \cite{2023structuralbalance}. Additionally, consensus implies identical values of magnitude as well as phase angles of the opinions within each cluster in our work, whereas states have equal magnitudes but unequal phase angles at steady state in \cite{2023structuralbalance}. Further, we summarize our key contributions below:
\begin{enumerate}
\item[i)] Exploiting the far-reaching applications of complex-weighted networks, we propose a novel interpretation to these networks by modeling \emph{emotions} using complex-weights in a social network. Building on this, we first propose a model that captures the process of \emph{emotional contagion}, which is based on the equitable partitions in complex-weighted networks. This leads to agents developing similar emotional states at steady state. 
\item[ii)] Secondly, we model the effect of similar emotions at steady state on the final opinions of agents in a social network. Our model reveals that emotionally similar users in a social group lead to similar final opinions, i.e., multi-partite consensus in complex-weighted digraphs. 
\end{enumerate}

To the best of our knowledge, this is the first study to incorporate emotions in the area of opinion dynamics and rigorously investigate the conditions to achieve \emph{multi-partite consensus} in complex-weighted networks.
\vspace{-2mm}
\section{Preliminaries}\label{sec:prelims}
\vspace{-1mm}
Let $\mathbb{C}^{n \times n}$ denote the space of $n\times n$ matrices with complex entries. For $n=1$, the space is $\mathbb{C}$, and let ${\Re}(m)$ and ${\Im}(m)$ denote the real and imaginary parts of $m \in \mathbb{C}$. Let $\iota$ be the imaginary number, where $\iota=\sqrt{-1}$. An identity matrix is denoted by $I$. 
The corank of a matrix is the complementary to its rank and refers to the dimension of the matrix's nullspace.   

 Let $\mathcal{G}(A)$ be signed when all its edge-weights are positive or negative.
Consider the signed (weighted) directed graph denoted as $\mathcal{G}=(\mathcal{V},\mathcal{E},A)$ where $\mathcal{V}$ is the set of vertices, $\mathcal{E}$ is the set of edges, and $A=[a_{ij}] \in \mathbb{R}^{n \times n}$ is the sociomatrix of the network $\mathcal{G}$. The non-zero entries of $A$ represent inter-personal relationships between individuals, where the positive and negative edges represent friendship and enmity, respectively. Similarly, $\mathcal{G}_{EC}=(\mathcal{V}_{EC},\mathcal{E}_{EC},A_{EC})$ is the graph showing the network structure of emotional contagion (EC) process, where $A_{EC} \in \mathbb{C}^{n \times n}$ is the adjacency matrix of $\mathcal{G}_{EC}$. Both the graphs have the same connectivity structure, but different edge-weights; the complex-weighted edges in $A_{EC}$ simultaneously represent both the intensity and the valence of emotional influence among agents.
Out-degree matrices are diagonal matrices where each diagonal entry represents the out-degree of a corresponding node. The out-degree of a node is equal to the sum of the edge-weights of all its outgoing edges \cite{saxena2025laplacian}. A node is a sink node if its out-degree is zero. Then, the Laplacian matrices are defined as $L=D-A$. We denote ${v}$ and ${w}$ to be the right and left eigenvectors corresponding to the zero eigenvalue of the signed or unsigned Laplacian matrix ($L$). Further, let ${v_e}$ and ${w_e}$ be the right and left eigenvectors corresponding to the zero eigenvalue of the complex-weighted Laplacian matrix ($L_e$). Let ${v_{EC}}$ and ${w_{EC}}$ be the right and left eigenvectors corresponding to the zero eigenvalue of the modified complex-weighted Laplacian matrix ($\tilde{L}=L_{EC}+f$). A cluster is defined as a subgraph consisting of a subset of agents in a digraph. 
 
Next, we present a few existing results, which will be used in the subsequent sections of our paper.
\begin{theory}\label{theory of positive emotions}({\textbf{Broaden-and-build theory of positive emotions \cite{2004broadentheory}})}
 Positive emotions expand an agent's awareness and urge to engage with the world, unlike negative emotions, which narrow our focus to specific actions. This broadened mindset helps us promote new ideas, and social bonds.
\end{theory}
\begin{definition}\label{def:EEP}({\textbf{Almost equitable partitions(AEP hereafter)\cite{liu2019graph}})}
 Suppose that \(\mathcal{G} = (\mathcal{V}, \mathcal{E})\) is a complex-weighted digraph and \(L\) is its Laplacian.  
A partition \(\pi\) of \(\mathcal{V} \setminus \{C_1, C_2, \ldots, C_r\}\) is said to be an AEP if for all \(s, t \in C_i, i, j = 1, \ldots, r\), the following equality holds:  
\begin{equation}
\sum_{k \in C_j} l_{sk} = \sum_{k \in C_j} l_{tk}.
\end{equation}
\end{definition}
\vspace{-2mm}
 \section{PROPOSED MODEL}\label{sec:proposed model}
 Researchers in the field of psychology have emphasized that emotions play a vital role in the evolution of opinions and the process of decision-making \cite{2015emotionanddecisionmaking,izard2009emotion}. {This is also evident from the existing empirical data, which shows that $57$ percent of the consumers make decisions based on their emotions\cite{{datacitation}}. For instance, people tend to make more purchases when they are happy. Consequently, emotions shape an individual's beliefs and behaviour, which in turn affects their judgement and decision-making. To understand and analyze the impact of emotions, we begin by characterizing and representing them. Several models have emerged over the years to systematically represent emotions.
  Among these, the \emph{Russell’s theory of emotions} stands out as one of the most widely accepted frameworks.} According to it, emotions depend on \emph{measures of valence} and \emph{intensity}, as shown in Fig.\ref{fig:circumplex}. Further, the emotional states of agents are represented on a circle and the model is also termed as \emph{circumplex model}. The origin in this circle represents the neutral state of emotions, x-axis represents the valence or pleasantness, ranging from sadness (negative valence) to happiness (positive valence), while the y-axis represents the degree of arousal, ranging from high arousal (alertness) to low arousal (calmness)\cite{russell1980circumplex}.
 \begin{figure}[H]
    \begin{subfigure}{0.28\textwidth}
        \centering
         \includegraphics[scale=0.1]{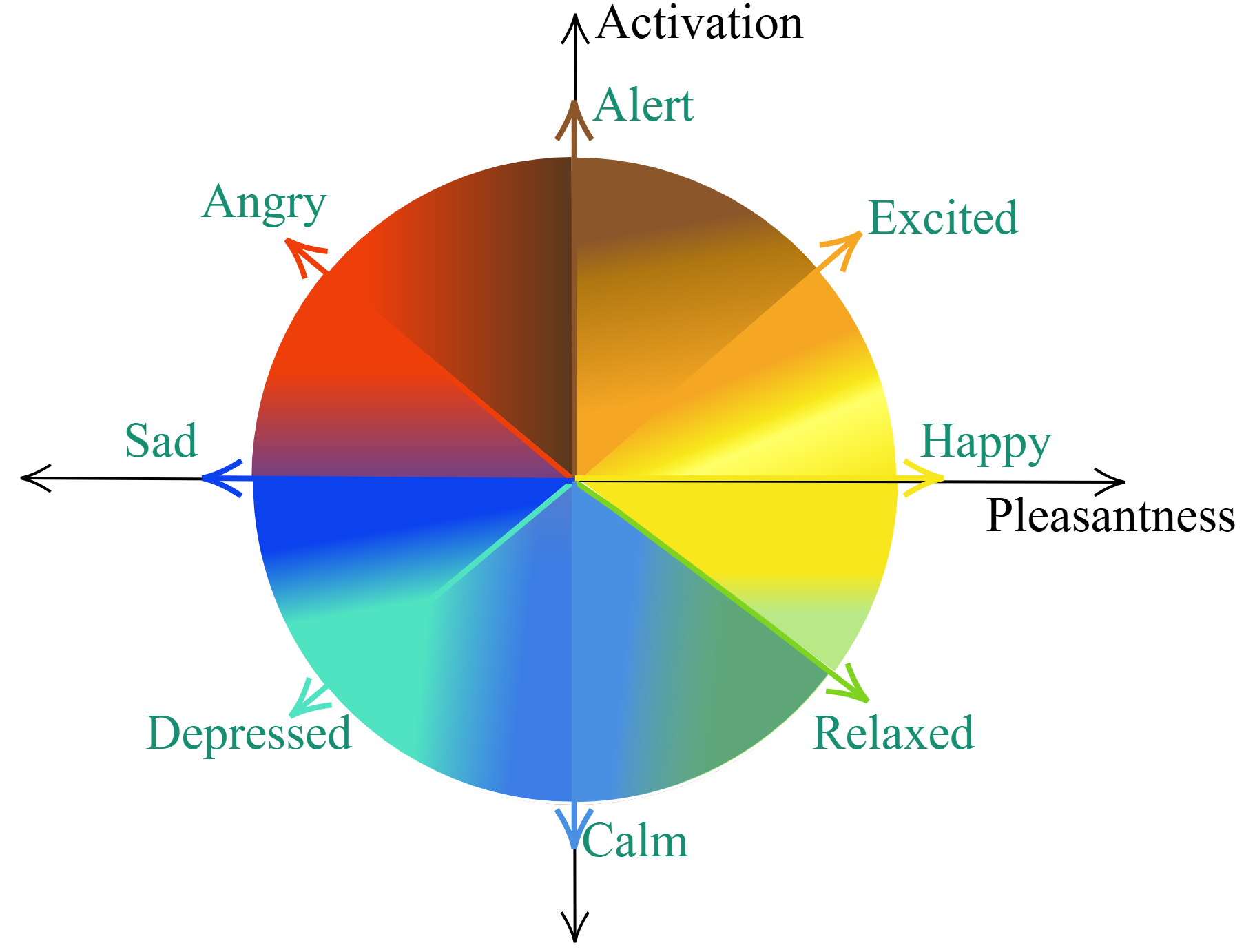}
 \caption{{\small}}
     \label{fig:circumplex}
        \end{subfigure}
    \hspace{-3mm}
    \begin{subfigure}{0.192\textwidth}
    \centering
   \includegraphics[scale=0.25]{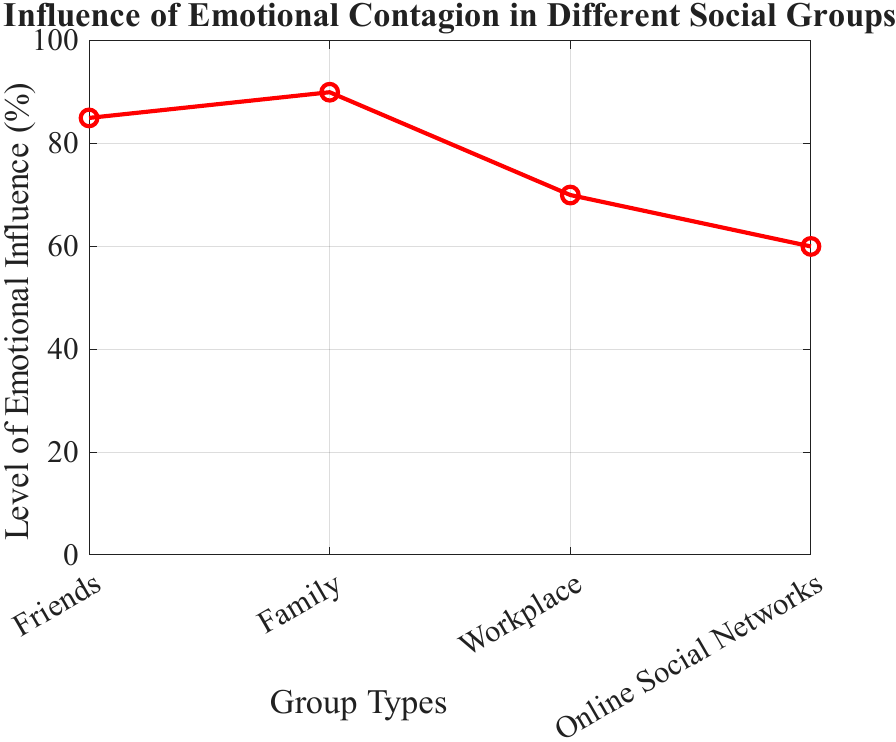}
        \caption{{\small }}
        \label{fig:emotionalcontagiongraph}
          \end{subfigure}     
          \raggedright
    \caption{\small (a) shows mapping of emotions with the color coding adopted from the movie `\emph{Inside Out}', and (b) shows the impact of emotional contagion in different social groups\cite{datacitation}}
    \label{}
\end{figure} 
Following the \emph{circumplex model}, we denote ${e} \in \mathbb{C}$ as the emotion of an individual as;
\begin{equation}\label{eq:emotioncartesianform}
  {e}=v+ \iota a,
\end{equation}
where $v$ and $a$ denote valence and arousal, respectively. This representation of emotions can also be done alternatively using a two-dimensional (2D) real-valued representation. But, we prefer the complex-valued representation because the 2D real-valued representations treat these dimensions independently, whereas a complex-valued representation combines them into a single variable while preserving their orthogonality through the real and imaginary components.
 
Having the mathematical representation of emotions, we now turn our attention to investigate the effect of emotions in the process of decision-making in social networks. {Emotions spread among agents as a result of \emph{emotional contagion} during social interactions and finally affect the decisions in a group of interacting agents. \emph{Emotional contagion} is the tendency to mimic and align expressions with your closed ones and converge emotionally \cite{herrando2021emotional}. Fig.\ref{fig:emotionalcontagiongraph} illustrates that \emph{emotional contagion} significantly influences the emotions of agents, and the \emph{extent of influence} depends on the nature of relationships between agents in social groups. 

Since \emph{emotional contagion} arises through repeated interpersonal interactions over a social network, it can be naturally modeled as a diffusion process on a digraph. The graph Laplacian provides a mathematically principled representation of pairwise influence and has been widely adopted in diffusion models due to its ability to capture the network topology and convergence properties\cite{peters2015multimodal,kikuchi2025network}. {Building on this, we propose the following Laplacian-based framework to capture the dynamics of emotions via \emph{emotional contagion} in a network:}
\begin{equation*} \dot{e} =-\Delta(D+f- A_{EC})e=-\tilde{L}e,
\end{equation*}
where $\tilde{L}$ is the modified Laplacian matrix with a modified degree matrix. 
The following notations are used to model the process of \emph{emotional contagion}:
\begin{itemize}
    \item $A_{EC} \in \mathbb{C}^{n \times n}$: adjacency matrix that encodes the network structure governing dynamics of \emph{emotional contagion}
    \item ${f},e $: a diagonal matrix consisting of feedback values for emotional modulation\footnote{Emotional modulation is an extrinsic mechanism applied to a subset of nodes in a network, represented mathematically as a perturbation to the Laplacian matrix(e.g., $\tilde{L} = L + f$), which results in a shift in the emotional states of the agents.} and a vector consisting of \emph{emotions}, respectively
    \item $\Delta \in \mathbb{R}^{n \times n}$: a diagonal matrix, consisting of $i^{\text{th}}$ diagonal entries with $\delta_i \in [0,1]$ as the \emph{emotional susceptibility} of the $i^{\text{th}}$ agent. 
\end{itemize}
The parameter, \emph{emotional susceptibility} quantifies how strongly an agent is influenced by the emotional states of its neighbors. A maximum value of $\delta_i = 1$ indicates maximum susceptibility, meaning the agent is fully influenced by its neighbors.
The digraph, $\mathcal{G}_{EC}$ is considered to be equitable partitioned; this implies that every agent within a cluster receives the same aggregate emotional influence from every other cluster. Consequently, every cluster is treated as a single emotional unit in further analysis. Following this, the feedback vector in our model ensures that the emotional states converge to steady state outcomes characterized by convergence such as consensus, or clustering of emotions. 

Motivated by the relevance of \emph{emotional contagion}, we further model its effect on the evolution of opinions in social networks. In this context, consider ${x}_{e,i}(t)$ as the $i^{\text{th}}$ entry of $x_e(t)$, which is the state vector consisting of opinions influenced by emotions.
Further, the existing inter-personal relationship between agents `$i$' and `$j$' is given by $a_{ij} \in \mathbb{R}$. Then, the inter-personal relationship is effected by the process of \emph{emotional contagion}, and is shaped by the relative emotional states of both the agents\cite{herrando2021emotional}. {We model the relative emotional influence by the term, $e_i/e_j$, where $e_i$ is the emotion of the $i^{\text{th}}$ agent at the steady state. The relative emotional influence is modeled such that agent $i$'s decision depends primarily on its own emotion with a weak influence on that of agent $j$'s final emotion. Additionally, the ratio ${e_i}/{e_j}$ preserves the same \emph{network structure} without any addition or deletion of edges and yields the same spectral structure for mathematical tractability.} The effective interpersonal relation between agents remains unchanged when they have same steady state emotional states, indicating that emotional influence arises only when there exists an emotional disparity between agents. 
 Combining the aforementioned facts, we propose the following model for an $i^{\text{th}}$ agent:
\begin{equation}\label{eq:modified model}
\dot{x}_{e,i}(t)
=
\sum_{j =1 }^n
a_{ij}\frac{e_i}{e_j}
\big(x_{e,j} - x_{e,i}\big).
\end{equation}
We assume that none of the agents in the network are emotionally stoic at the steady state, i.e., $e_j \neq 0$ for all the agents.
Then, we present the combined model of \emph{emotional contagion} and its effect on final opinions in compact form in Eqn.\eqref{eq:combined_model} as; 
\begin{equation}
\begin{bmatrix}
\dot{e} \\ 
\dot{x}_e
\end{bmatrix}
=
\begin{bmatrix}
-\Delta(D+f-A_{EC}) &  0 \\[2pt]
0 & { -L_e}
\end{bmatrix}
\begin{bmatrix}
e \\ 
x_e
\end{bmatrix},
\label{eq:combined_model}
\end{equation}
where $L_e \in \mathbb{C}^{n \times n}$ is the {Laplacian matrix where the relationships are affected by steady state emotions.}
 From a psychological perspective, the proposed model formalizes the idea that emotional alignment between agents enhances receptivity or willingness to agree to others' opinions, whereas emotional misalignment suppresses such responsiveness in an ongoing discussion (see Table \ref{table:emotions}). Furthermore, we analyze our model and derive the main results in the next section.
 \vspace{-2mm}
\section{EMOTIONAL CONTAGION AND ITS EFFECT ON SOCIAL NETWORKS}
Among networked systems, social networks stand out as one of the most relatable and intricate systems. The social networks can be co-operative or antagonistic, depending on the relationships between agents. 
We propose that complex-weighted networks arise naturally when the emotions of agents are involved in the process of decision-making. 
Further, in order to analyze \emph{multi-partite consensus} of emotions and opinions, we introduce the following assumption, which is imposed on all the networks considered in this work.
\begin{assumption}\label{assumption1}
The network of $N$ agents showing the emotional influences is described by a complex-weighted digraph $\mathcal{G}_{EC}$, with $\mathcal{V} = [1,N]$. The corresponding digraph is almost equitably partitioned, in $k \geq 2$ clusters, $p\in [1, k]$, along with one sink node in the digraph. 
\end{assumption}
{Next, we provide the convergence analysis of the proposed model, given in Eqn.\eqref{eq:combined_model}. Throughout this work, we assume that the networks adhere to Assm.\ref{assumption1}. 
\begin{theorem}\label{thm:consensus}
Let $\mathcal{G}$ be a signed digraph, where the corresponding Laplacian matrix, $-L \in \mathbb{R}^{n \times n}$ is marginally stable of corank $1$. Let $\mathcal{G}_{EC}$ be a complex-weighted digraph satisfying Assm.\ref{assumption1} and $\Re({v_{EC}}) >0,~\Re({w_{EC}}) >0$ be the right and left eigenvectors corresponding to the zero eigenvalue of $\tilde{L}$. Then, according to the proposed model in Eqn.\eqref{eq:combined_model}, the emotions and opinions converge to consensus if and only if all the agents in $\mathcal{G}_{EC}$ have $f_i = 0, \forall ~i \text{ and the matrix,} - \tilde{L}$ is marginally stable of corank $1$.
\end{theorem}
\begin{proof}
\textbf{(If)} We have, $f_i = 0 ~ \forall ~i \in \{ 1,2,..,n\}$. Then, it is also given that the matrix, $- \tilde{L}=-\Delta (D-W_{EC})=-\Delta {L_{EC}}$ is marginally stable of corank $1$. 
Further, the right and left eigenvectors corresponding to the simple zero eigenvalue are; ${v_{EC}} = K\mathbb{1}_{n}$, where $K \in \mathbb{C}$ and ${w_{EC}}$, respectively. The final state vector is given by ${e}_{f} = \exp(-\Delta L_{EC}t) {e}_{0}={v_{EC}} {w_{EC}}^{H} {x}_{0}= ({w_{EC}}^{H} {x}_{0})K\mathbb{1}_{n}$. Thus, the final emotional states of agents converge to a consensus, provided the product of the left eigenvector and initial conditions is not zero. 
Consequently, the diagonal matrix consisting of emotional states, $E$ can be written as $E=e_i I_n$ at the steady state.

\textbf{(Only if)} We have, $E=e_i I_n$ at the steady state, which is only possible when ${v_{EC}} = K\mathbb{1}_{n}$, which requires $f_i = 0 ~ \forall ~i \in \{ 1,2,..,n\}$ and $- \tilde{L}$ is marginally stable of corank $1$.

Furthermore, we have $-L \in\mathbb{R}^{n \times n}$ is marginally stable of corank $1$, and $L\mathbb{1}_n=\mathbb{0}_n$ (follows from the definition of Laplacian matrix in Sec.\ref{sec:prelims}) and the resultant Laplacian matrix affected by steady state emotions is given by $L_e = E L E^{-1}$. Using the similarity transformation, the spectrum of $L \text{ and } L_e$ are same, which ensures the marginal stability of $L_e \in \mathbb{C}^{n \times n}$ for complex-weighted digraphs. This implies that $L_e$ has a simple zero eigenvalue with the non-zero eigenvalues in the ORHP of the s-plane. 
Following this, the steady state solution of the transformed flows are given by Eqn.\eqref{eq:zdomainsteadystate1};
\begin{equation}\label{eq:zdomainsteadystate1}
    \underset{t \to \infty}{lim} ~x_e(t)=  \underset{t \to \infty}{lim} ~\exp({-L_e t})x_e(0)={v_e} {w_e}^H x_e(0),
\end{equation}
where ${v_e}$ and ${w_e}$ are the corresponding right and left eigenvectors of zero eigenvalue in the spectrum of $L_e$.
Next, the resultant right eigenvector corresponding to the zero eigenvalue is ${v_e}=E \mathbb{1}_n$, where $E=e_i I_n$ from the previous part of the proof. This follows from the similarity transformation.
\end{proof}
Therefore, it is evident that {consensus of emotions at steady state} leads to consensus of opinions in the resultant complex-weighted networks. We obtain \emph{complex-valued opinions} at the steady state, which represent the influence of emotions, allowing agents to express opinions with varying levels of support or opposition rather than being limited to completely affirmative or negative opinions.

Next, we introduce the characteristic matrix, $P \in \mathbb{R}^{n \times k}$ for AEP digraphs, which encodes the cluster membership of each node (see Ex.\ref{ex:multipartite}). Further, we define the quotient matrix as $L^{\pi}=(P^TP)^{-1} P^TLP$, where $L^{\pi} \in \mathbb{C}^{k \times k}$ for $k$ number of equitable partitions in the AEP digraphs \cite{liu2019graph}. Using these, we present the following result, which will be used to ensure clustering in complex-weighted networks (in Thm.\ref{thm:multipartite}).
\begin{lemma}\label{lemma:AEP}
 Let $\pi$ be an equitable partition with characteristic matrix, 
$P$. Then, the corresponding quotient matrix $L^{\pi}$ satisfies the invariance condition $\tilde{L} P=P L^{\pi}$.
\end{lemma}
\begin{proof}
The proof follows on the same argument as in \cite{liu2019graph}, because the complex-weighted networks are AEP in the sense of Def.\ref{def:EEP}. It is omitted in the interest of space.
\end{proof}

Now, we present conditions to achieve \textit{multi-partite consensus,} which is the generalization of \emph{consensus} and \textit{bipartite consensus}. There are multiple clusters, where the agents within each cluster converge to the same opinion. Each cluster here represents a distinct real part as well as imaginary part at the steady state in the process of opinion evolution. 
\begin{theorem}\label{thm:multipartite}
Let $\mathcal{G}(A)$ be a signed digraph, where the corresponding Laplacian matrix, $-L \in \mathbb{R}^{n \times n}$ is marginally stable of corank $1$. Let $\mathcal{G}(A_{EC})$ be a complex-weighted digraph satisfying Assm.\ref{assumption1}, $\Re({v_{EC}}) >0,~\Re({w_{EC}}) >0$ be the right and left eigenvectors corresponding to zero eigenvalue of $\tilde{L}$. Then, the proposed model in Eqn.\eqref{eq:combined_model} exhibits clustering of emotions and opinions if the values of $f_p \text{ in } f =  \bigoplus_{p=1}^{k}f_{pp}  ={diag}(f_p)$ are chosen such that the matrices $\tilde{L}$ and $L^{\pi}$ are marginally stable of corank $1$.
\end{theorem}
\begin{proof} 
We have the perturbed matrix ($\tilde{L}\in \mathbb{C}^{n \times n}$), and its quotient matrix, $L^{\pi}\in \mathbb{C}^{n \times n}$, is marginally stable of corank $1$.
First, we show by contradiction that due to the modified degree matrix ($D=\operatorname{diag}(\tilde{d}_{ii} ); \tilde{d}_{ii} = \sum_{j=1}^N a_{ij} + f_p ~ \text{for } i \in \mathcal{V}_p$); the matrix, $\tilde{L}$ as well as $L^{\pi}$ does not have the eigenvector corresponding to zero eigenvalue in the $\operatorname{span}{\{\mathbb{1}_n\}}$.

    Assume $v_{EC}=K\mathbb{1}_n$ is an eigenvector of $\tilde{L}$, then
\[
\tilde{L} K\mathbb{1}_n = \mathbb{0}_n \text{ for }K \in\mathbb{C},
\]
\[
\tilde{L}K\mathbb{1}_n = ({L} + f)K\mathbb{1}_n = fK\mathbb{1}_n = K(f_1, \dots, f_n)^T.
\]
This contradicts the assumption that $v_{EC}=K\mathbb{1}_n$. Therefore, $v_{EC} \notin \operatorname{span}\{\mathbb{1}_n$\}.

Secondly, it is sufficient to prove that the eigenvector corresponding to the zero eigenvalue is of the form $v_{EC}=[\alpha_1 \mathbb{1}_{n1} ~ \alpha_2 \mathbb{1}_{n2}~ \alpha_3 \mathbb{1}_{n3}~...~\alpha_k \mathbb{1}_{nk}]^T$. 
We provide the following proof for the same: 
\begin{equation*}
    L^{\pi}v^{\pi}= \lambda v^{\pi}
\end{equation*}
Multiplying with the matrix, $P$ both sides, we get
\begin{equation*}
   P L^{\pi}v^{\pi}= \lambda P v^{\pi}
\end{equation*}
Using Lemma \ref{lemma:AEP}, and substituting $P L^{\pi}=\tilde{L}P$;
\begin{equation*}
  \tilde{L}Pv^{\pi}= \lambda P v^{\pi}
\end{equation*}
\begin{equation*}
 \tilde{L}(Pv^{\pi})= \lambda (P v^{\pi})
\end{equation*}
Thus, the eigenspaces of matrices, $L^{\pi}$ and $\tilde{L}$ are related to each other.
Thus, this implies $v_{EC}=[\alpha_1 \mathbb{1}_{n1} ~ \alpha_2 \mathbb{1}_{n2}~...~\alpha_k \mathbb{1}_{nk}]^T$, depending on the values of characteristic matrix, $P$. 

Further, the steady state solution for opinions of the transformed flows are given by Eqn.\eqref{eq:zdomainsteadystate};
\begin{equation}\label{eq:zdomainsteadystate}
    \underset{t \to \infty}{lim} ~x_e(t)=  \underset{t \to \infty}{lim} ~\exp({-L_e t})x_e(0)={v_e} {w_e}^H x_e(0),
\end{equation}
where ${v_e}$ and ${w_e}$ are the corresponding right and left eigenvectors of zero eigenvalue in the spectrum of $L_e$.
Next, the resultant right eigenvector corresponding to the zero eigenvalue is ${v_e}=E \mathbb{1}_n$, where \[E = \bigoplus_{p=1}^{k} e_{pp}, \qquad e_{pp} = \operatorname{diag}(e_p),\]
where $e_p \in \mathbb{C}$ denotes the {steady state emotion} of the $p^{\text{th}}$ agent. Finally, due to the structure of matrix, $E$ and Eqn.\eqref{eq:zdomainsteadystate}, the final opinions converge to multi-partite consensus.
\end{proof}\vspace{-2mm}
\begin{remark}
The afore-mentioned results extend analogously to unsigned digraphs as well. The proposed framework is fully distributed, where each agent updates its emotional state based on its local interactions. Such distributed models are well suited for capturing decentralized social interactions and opinion dynamics in large-scale networks. {However, the present analysis is based on the steady state emotional states and does not address the coupled time-varying opinion dynamics. We aim to extend our framework to time-varying network topologies.}
\end{remark}
 

 Finally, we establish that all the emotionally-similar agents converge to consensus and the whole network converges to multi-partite consensus.
Next, we present the following numerical example to validate the aforementioned results. 
\begin{example}\label{ex:multipartite}
Let $L_{EC}$ be the Laplacian matrix corresponding to the digraph in Fig.\ref{fig:3clusterdigraph} partitioned into three clusters:
\begin{figure}[h!]
\centering
\includegraphics[scale=0.22]{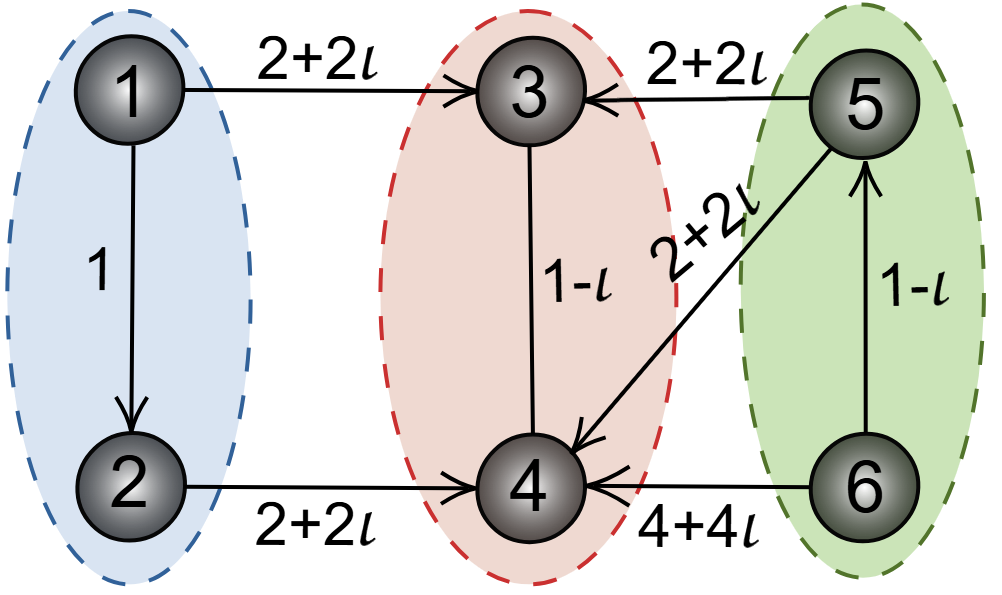}
\caption{An AEP complex-weighted digraph}
\label{fig:3clusterdigraph}
\end{figure}
\end{example}\vspace{-2mm}
The modified Laplacian matrix is given by $\tilde{L}=(D+f)-W_{EC}$, where $f=diag{(10-6\iota, 10-6\iota, 0,0,5+2\iota,5+2\iota)}$. The corresponding Laplacian matrix is marginally stable with corank $1$ with $v_{EC} \notin \operatorname{span}\{\mathbb{1}_n\}$.
Further, according to Thm.\ref{thm:multipartite}, the transpose of characteristic matrix $P^T$ is given as;
\begin{equation*}
P^T=\left[
\begin{array}{@{}lllllll@{}}
1   & 1    & 0   & 0  & 0  & 0  \\
0   & 0    & 1 & 1 & 0  & 0     \\
 0   & 0    & 0 & 0 & 1  & 1     
 \end{array}
\right],
\end{equation*}
 and thus, the quotient matrix is $L^{\pi}=(P^TP)^{-1} P^T\tilde{L}P$ has $v_{\pi}=[ 0.087 +0.18, ~0.86,~0.45 + 0.09\iota]^{\mathrm{T}}$. Following this, $v_{EC}=Pv_{\pi}$ and thus, the Laplacian flows converge to multi-partite consensus, as shown in Fig.\ref{fig:multipartitesim}.
\section{SIMULATION RESULTS}
In this section, we demonstrate the validity of our theoretical results using a synthetic social network and a real-world dataset. We begin by simulating the AEP digraph presented in Ex.\ref{ex:multipartite}. The network is shown in Fig.\ref{fig:3clusterdigraph}, and the dynamics of emotional contagion are simulated accordingly. As predicted by Thm.\ref{thm:multipartite}, the emotional states converge to distinct emotions depending on the feedback matrix. This is clearly validated in Fig.\ref{fig:multipartitesim}. 
\begin{figure}[h!]
    \centering
    \includegraphics[scale=0.56]{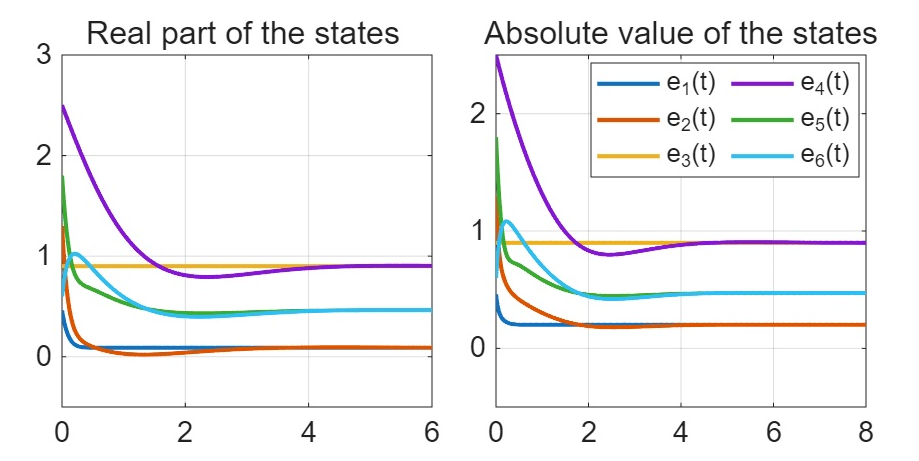}
    \caption{Emotional states of agents converging to clustering (or multi-partite consensus) in the digraph discussed in Ex.\ref{ex:multipartite}}
    \label{fig:multipartitesim}
\end{figure}
\subsection*{Zachary Club's Network} \vspace{-1mm}
Zachary’s Karate Club Network consists of $34$ agents (nodes) and $78$ friendship ties (edges), representing social interactions observed over approximately two years. During the course of the study, a conflict emerged between the club’s administrator and the instructor, leading to a division of the club into two factions. Following this, $18$ members aligned with the administrator, and the remaining $16$ members joined the instructor’s faction. 
Since it is an unsigned network, the corresponding Laplacian matrix $(L)$ is marginally stable with corank $1$. We modify the sociomatrix according to Eqn.\eqref{eq:modified model} and simulate the modified Laplacian flows using the $E$ matrix to show different behaviors with different emotions in Fig.\ref{fig:zacharysim}. Here, we only show the evolution of opinions based on the standard partitions of the graph, since the graph is not AEP.
\begin{figure}[h!]
    \centering
    \includegraphics[scale=0.53]{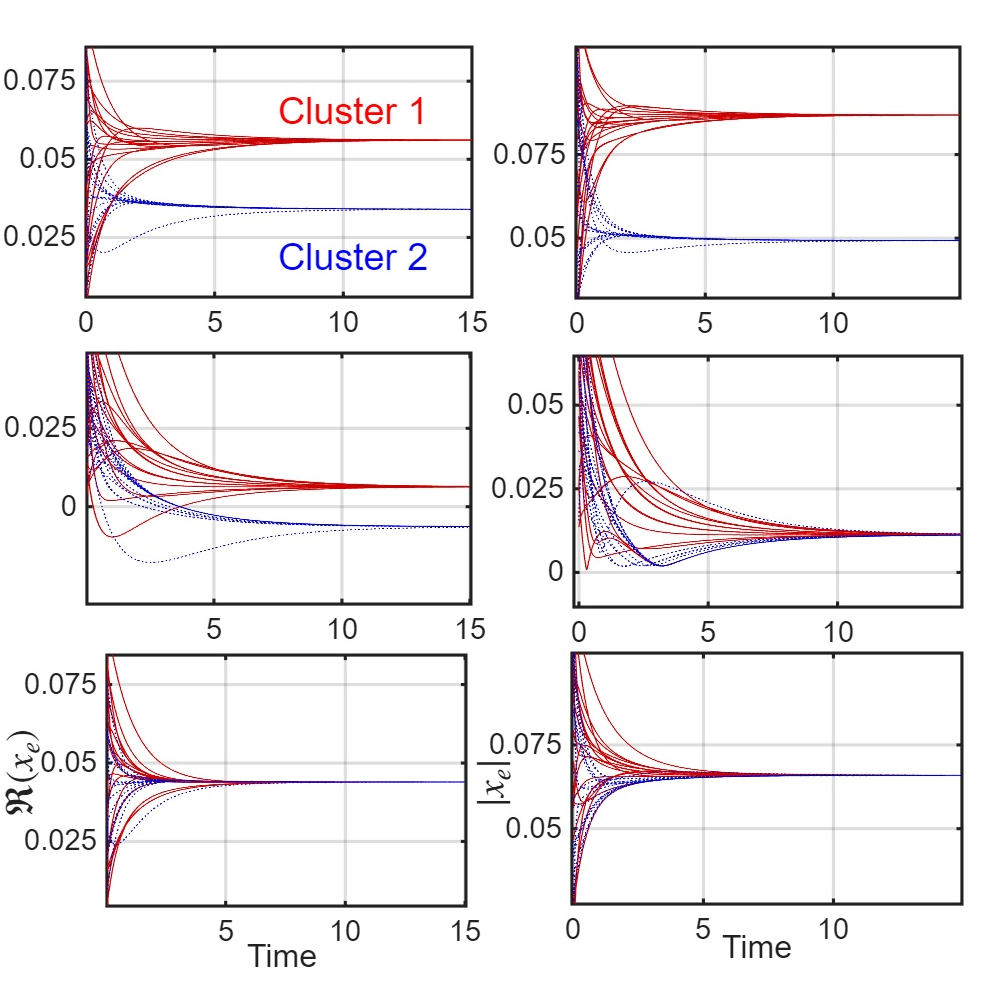}
    \caption{\small Plots illustrating the behavior of agents in Zachary’s Karate Club network, showing the final opinions of the two factions to clustering, polarization, and consensus, respectively}
    \label{fig:zacharysim}
\end{figure}\vspace{-4mm}
\section{DISCUSSION}\label{sec:discussion}
Our findings reveal important insights into how emotional contagion occurs and affects opinion dynamics, with meaningful implications for understanding collective emotional states in social networks. 
Specifically, we show the effect of emotional contagion on the steady state opinions, where agents in a cluster lead to emotional convergence within a social network. As agents emotionally converge, their opinions tend to align with each other, converging to consensus within clusters or, polarization across groups in some cases. 
For instance, we observe that a social group leads to polarization when two of its factions hold completely opposite emotions in the Zachary Club's network (see Fig.\ref{fig:zacharysim}).
Moreover, a significant application of our study lies in its practical implications for improving mental health and emotional well-being in social groups. As per our results, emotions are not merely individual feelings but collective phenomena that can be varied through systematic intervention (modeled using a feedback matrix in our model). Intuitively, an agent in a distressed state surrounded by agents having positively valenced emotional states may, over time, experience emotional uplift through sustained exposure to positive valenced emotions like happiness. 
To summarize, we present a table\footnote{In the interest of brevity, we present four primary emotions, noting that others can be expressed as functions of these emotions.} (Table \ref{table:emotions}), which shows the effect of emotions on the relationships as per our proposed model. 
\begin{table}[h!]
\centering
\caption{Behavior of the Proposed Model}\label{table:emotions}
\renewcommand{\arraystretch}{1.1}
\small
\setlength{\tabcolsep}{4pt}
\begin{tabular}{cccc}
\toprule
$e_i$ & $e_j$ & Friends ($|e_i| > |e_j|$) & Enemies ($|e_i| > |e_j|$) \\
\midrule
Happy & Happy & Co-operation $\uparrow$ & Antagonism $\downarrow$ \\
Happy & Sad   & Antagonism$^{\ast}$ & Co-operation$^{\ast}$$^{\ast}$\\
Happy & Calm  & Unbiased & Unbiased \\
Happy & Alert & Unbiased & Unbiased \\
Sad   & Sad   & Co-operation $\uparrow$ & Antagonism $\downarrow$ \\
Sad   & Calm  & Unbiased & Unbiased \\
Sad   & Alert & Unbiased & Unbiased \\
\bottomrule
\end{tabular}\\
\end{table}

According to Table\footnote{
$^{\ast}${Co-operation is converted to antagonism} \\
$^{\ast}$$^{\ast}${Antagonism is converted to co-operation}} \ref{table:emotions}, the emotional state of happiness increases co-operation and decreases antagonism among the agents. This is in line with the theory of positive emotions (see Theory \ref{theory of positive emotions}). Instead, an agent in a sad emotional state converts co-operation to antagonism and vice-versa. However, it is worth noting that when both the agents are sad, they tend to co-operate with each other, as shared sadness promotes mutual empathy. Further, when the agents are calm or alert, we infer that the decisions are not significantly affected by the existing relationship (co-operation or antagonism). Finally, the modified edge-weight, $(a_e)_{ij}$ implying complex-valued opinions later, can be expressed as $(a_e)_{ij}=\Re((a_e)_{ij}) + \Im ((a_e)_{ij})$, where $\Re((a_e)_{ij}) $ signifies social alignment, whereas  $\Im((a_e)_{ij})$ signifies emotional control between a pair of agents.
\section{CONCLUSION \& FUTURE SCOPE}\label{sec:conclusions}
{This paper proposes a continuous-time linear time-invariant model that demonstrates the spread of emotions and incorporates the influence of final emotions in the decision-making process.} Using this model, we show that when the agents share the same emotional state, they converge to a common complex-valued opinion. Recent studies in the field of neuroscience show that this behavior is a consequence of the mirror neuron system, which aligns with our results\cite{2022decreasingalertness}. 
The proposed framework is novel in the sense that the classical opinion dynamics models provide a strong mathematical foundation for interaction processes, but do not explicitly account for the role of emotions in shaping opinion evolution. This framework further enables us to investigate more general forms of agreement beyond consensus; in particular, we establish conditions for achieving multi-partite consensus in complex-weighted networks. To achieve multi-partite consensus, we introduce the concept of AEPs in complex-weighted networks. 
{We plan to analyze the effect of multiple emotions and explore time-varying complex models in our future work.}
\vspace{-2mm}
\bibliography{References}
\bibliographystyle{unsrt}
\end{document}